\documentclass[16pt]{article}
\usepackage{graphicx,amssymb,amsmath,amsthm}
\usepackage{comment,cite,color}
\usepackage{cite,color}
\usepackage{algorithmic}
\usepackage{mathrsfs}
\usepackage[ruled]{algorithm}
\usepackage{epsfig}

\newtheorem{thm}{Theorem}
\newtheorem{cor}{Corollary}
\newtheorem{lem}{Lemma}

\theoremstyle{definition}
\newtheorem{defn}{Definition}

\newtheorem{remark}{Remark}

\newcommand{\R}{\mathbb{R}}

\newcommand{\RIP}{{\rm RIP}}
\newcommand{\OMP}{{\rm OMP}}
\newcommand{\supp}{{\rm supp}}
\newcommand{\aarg}[1]{\underset{#1}{\rm argmin}}
\date{}

\begin{document}
\bibliographystyle{plain}
\title{A remark about orthogonal matching pursuit algorithm  }
\author{  Zhiqiang Xu\thanks{Supported by the National Natural Science Foundation of China (10871196).}}
 \maketitle

\begin{abstract}
In this note,  we investigate the theoretical properties of Orthogonal Matching Pursuit (OMP), a class of decoder  to recover sparse signal in compressed sensing. In particular, we show that the OMP decoder can  give
$(p,q)$ instance optimality for a large class of encoders with $1\leq p\leq q
\leq 2$ and $(p,q)\neq (2,2)$. We also show that, if the encoding matrix is
drawn from an appropriate distribution, then the OMP decoder is $(2,2)$ instance
optimal in probability.
\end{abstract}

\section{Introduction}
We consider a signal $x\in \R^N$ where $N$ is large and denote by
$\Sigma_k$ the set of $k$-sparse vectors, i.e.,
$$
\Sigma_k:=\{x\in \R^N : \# {\rm supp}(x)\leq k\},
$$
where ${\rm supp}(x)$ is the set of $i$ for which $x_i\neq 0$ and
$\# A$ is the number of elements in the set $A$. Given a norm
$\|\cdot\|_X$ on $\R^N$, we set
$$
\beta_k(x):= \aarg{z\in \Sigma_k}\|x-z\|_X,
$$
and
$$
\sigma_k(x)_X:= \|x-\beta_k(x)\|_X,
$$
and call $\beta_k(x)$ and $\sigma_k(x)_X$ as the {\em the best k-term approximation} and {\em the best k-term approximation
error}, respectively.

 In Compressed Sensing theory, the information we gather about $x$
 can be described by
 $$
 y= \Phi x
 $$
 where $\Phi$ is an $n\times N$ matrix.  To recover $x$ from
 $y$, we use a decoder $\Delta:\R^n\rightarrow \R^N$ and denote $x^*:=\Delta(y)=\Delta(\Phi
 x)$. The $x^*$ can be considered as an approximation of $x$. Following Cohen, Dahmen and DeVore, we  say that a pair $(\Phi,\Delta)$ is {\em instance optimality} of order $k$
 with constant $C$ for the norm $X$ if it satisfies
 \begin{equation}\label{eq:ins}
 \|x-\Delta(\Phi x)\|_X \leq C \sigma_k(x)_X,\quad\quad  \text{ for all } x \in \R^N.
 \end{equation}
In general, one  chooses $X$ as $\ell_p$ (quasi-)norm where $p>0$. To state conveniently, throughout of this paper, we use the subscript $p$ to denote the $\ell_p$ norm.  Using  the notation of \cite{notation}, we say that a pair $(\Phi,\Delta)$ is $(q,p)$ instance optimality of order $k$ with constant $C$ if
$$
 \|x-\Delta(\Phi x)\|_{q} \leq C \frac{\sigma_k(x)_{p}}{k^{1/p-1/q}},\quad\quad  \text{ for all } x \in \R^N.
$$
 The theoretical analysis of instance optimality is presented in \cite{jams1}.

We next introduce a class of encoding matrix.  Following Cand\`{e}s and Tao, we say that the matrix $\Phi$ satisfies the Restricted Isometry Property (RIP) of order $k$ and constant $\delta_k \in (0,1)$ if
\begin{equation}\label{eq:con}
(1-\delta_k) \|x\|_2^2 \leq \|\Phi x\|_2^2 \leq (1+\delta_k) \|x\|_2^2
\end{equation}
holds for all $x\in \Sigma_k$. Throughout the rest of the paper, using the notation of \cite{notation}, we say that the matrix $\Phi$ satisfies $\RIP(k,\delta)$ if $\delta_k<\delta$.

 In the past, one investigated  the instance optimality of $\ell_1$ minimization, which is given by
 $$
 \Delta_1^\epsilon(y):=\aarg{x}\|x\|_{1}, \quad\quad \text{ subject to } \|\Phi x-y\|_2\leq \epsilon.
 $$
In \cite{candes}, Cand\`{e}s improved on the work of Cand\`{e}s, Romberg and Tao \cite{tao} and showed that
$$
\|\Delta_1^\epsilon(\Phi x)-x\|_2\leq C_0\sigma_k(x)_1/\sqrt{s}+C_1\epsilon
$$
provided $\Phi$ satisfies $\RIP(2k,\sqrt{2}-1)$. The result implies that, in the noise-free case, i.e., $\epsilon=0$,
$(\Phi,\Delta_1^0)$  is $(2,1)$ instance optimality of order $k$ if $\Phi$ satisfies $\RIP(2k,\sqrt{2}-1)$.  In \cite{lp}, Saab and Y{\i}lmaz extended the results to $\ell_p$ decoder, where $0<p<1$, which is defined by
$$
\Delta_p^\epsilon(y):=\aarg{x}\|x\|_p\quad\quad \text{subject to } \|\Phi x-y\|_2\leq \epsilon,
$$
and showed that $(\Phi,\Delta_p^0)$ is $(2,p)$ instance optimality of order $k$ if $\Phi$ satisfies some RIP condition (see \cite{lp}). Independently, Foucart and Lai \cite{lai} also proved that $(\Phi,\Delta_p^0)$ is $(2,p)$ instance optimality under other  sufficient conditions.

In compressed sensing, an alternative decoder  is Orthogonal Matching Pursuit (OMP).  The major advantages of OMP are its ease of implementation  and potentially faster than $\Delta_1^\epsilon$ (see \cite{COMP,tropp1,tropp2}). However, so far,  very few theoretical results about the instance optimality of  OMP are known.  The aim of the note is the investigation of  the instance optimality property of  OMP decoder. To state conveniently, we  use $\OMP_M$ to denote the OMP decoder with $M$ iterations (see Algorithm 1). Then combining the methods developed by Cohen, Dahmen and DeVore \cite{jams1} and the result obtained by Zhang \cite{zhang}, we can prove the following result, which is the main result of this note:

\begin{thm}\label{th:lplq}
Suppose that $1\leq p \leq q \leq 2 $ and $0<\delta \leq 1$. We furthermore suppose that $p\neq 2$. Let $\Phi$ be any matrix which satisfies RIP condition
$\RIP(L,\delta)$ and $\delta_k+(1+\delta)\delta_{\alpha k} \leq \delta$, where $\alpha:=\lceil16+15\delta \rceil$  and $
 L:=k\left({N}/{k}\right)^{2-{2}/{p}}$. Then for any signal $x$ and any permutation $e$ with $\|e\|_2\leq \epsilon$, the solution  $x^*:=\OMP_{2(\alpha-1)k}(\Phi x+e)$ obeys
\begin{equation}\label{eq:noise}
\|x^*-x\|_q\,\, \leq\,\, C_0\frac{\sigma_k(x)_p}{k^{1/p-1/q}}+C_1 k^{1/q-1/2} \epsilon,
\end{equation}
where  $C_0=1+C_1+(2\alpha)^{1/q-1/2}$ and $C_1=(2\alpha)^{1/q-1/2}(2(1+\delta)(\sqrt{11+20\delta}+1)+1)$.
\end{thm}

\begin{algorithm}
\begin{algorithmic}
 \STATE {\bf Input:} encoding  matrix $\Phi$, the vector $y$, maximum allowed sparsity $M$
 \STATE {\bf
Output:} the $x^*$.
 \STATE {\bf Initialize:} $r^0=y,c^0=0,\Lambda^0=\emptyset, \ell=0$.
  \WHILE {  $\ell< M$}
  \STATE {\bf match:} $h^\ell=\Phi^Tr^\ell$
  \STATE{\bf identity:} $\Lambda^{\ell+1}=\Lambda^\ell\cup\{{\rm argmax}_j|h^\ell(j)|\}$
  \STATE{\bf update:}  $c^{\ell+1}=\aarg{z: {\rm supp}(z)\subset \Lambda^{\ell+1}}\|y-\Phi z\|_2$
  \STATE \,\,\qquad\qquad $r^{\ell+1}=y-\Phi c^{\ell+1}$
  \STATE \,\,\qquad\qquad $\ell=\ell+1$
 \ENDWHILE
 \STATE \,\, $x^*=c^{M}$
\end{algorithmic}
\caption{\small{$\OMP_M(y)$}}
\end{algorithm}

 By setting $(q,p)=(2,1)$ and $\epsilon=0$ in Theorem \ref{th:lplq}, we obtain the following Corollary:
\begin{cor}
Suppose that $\Phi$ satisfies the RIP condition   $\delta_{2k}+(1+\delta)\delta_{2\alpha k} \leq \delta $. Then
 $$
 \|\OMP_{2(\alpha-1) k}(\Phi x)-x\|_2 \leq C_2 \sigma_k(x)_1/\sqrt{k},
 $$
 where   $\alpha=\lceil16+15\delta \rceil$ and $C_2={2(1+\delta)}(\sqrt{11+20\delta}+1)+3 $.
\end{cor}

\begin{remark}
Theorem \ref{th:lplq} implies that $(\Phi, \OMP)$ is $(q,p)$ instance optimality of order $k$ provided $\Phi$ satisfies RIP condition of order $k(N/k)^{2-2/p}$.  Note that $n\times N$ matrix $\Phi$ can have RIP of order $\tilde{k}$ if $\tilde{k}=O(n/\log(N/n))$. Then, for OMP decoder,   $(q,p)$ instance optimality can be achieved at the price of $O(k(N/k)^{2-2/p}\log(N/k))$ measurements.
As shown in \cite{jams1} (Theorem 7.3), the number of measurements is optimal
up to a constant.
\end{remark}
\begin{remark}
No recover method can improve the term $k^{1/q-1/2} \epsilon$ on the right side of (\ref{eq:noise}) for arbitrary perturbations $e$. To see why this true, suppose that we take $x\in \Sigma_k$ and we know in advance the support of $x$, i.e., $T_0=\supp(x)$. Using this additional information, as shown in \cite{tao}, Least-Square has the best performance to recover $x$. Set $y=\Phi x+e$. Then
$$
x^*=
\begin{cases}
(\Phi_{T_0}^t\Phi_{T_0})^{-1}\Phi_{T_0}^ty, & \hbox{ on $T_0$ },\\
0,  & \hbox{ elsewhere }.
\end{cases}
$$
A simple observation is that  $(x^*-x)_{T_0^c}=0$ and
$$
(x^*-x)_{T_0}=(\Phi_{T_0}^t\Phi_{T_0})^{-1}\Phi_{T_0}^te.
$$
Here, we use $T_0^c$ to denote the complement of $T$ and $x_{T_0}$ to denote the vector which agrees with $x$ on $T_0$ and has all components equal to zero on $T_0^c$. Then H\"{o}lder inequality implies that
$$
\|x^*-x\|_q=\|(\Phi_{T_0}^t\Phi_{T_0})^{-1}\Phi_{T_0}^te\|_q \leq k ^{1/q-1/2}
\|(\Phi_{T_0}^t\Phi_{T_0})^{-1}\Phi_{T_0}^te\|_2,
$$
where the equality holds for some non-zero permutation $e$. Since $\Phi$ satisfies RIP condition of order $L$,
$$
\|(\Phi_{T_0}^t\Phi_{T_0})^{-1}\Phi_{T_0}^te\|_2 \approx \| \Phi_{T_0}^te\|_2 \approx \epsilon,
$$
which implies that $\|x^*-x\|_q\approx k^{1/q-1/2}\epsilon$ for some non-zero permutation $e$ provided $x\in \Sigma_k$.
\end{remark}
\bigskip

We next consider the $(2,2)$ instance optimality. As pointed out in \cite{jams1}, to obtain $(2,2)$ instance optimality with order $k=1$,
 one has to require the number of measurements is  $O(N)$, which is not what we hope.
 Hence, {instance optimality in probability} is the proper formulation in $\ell_2$. We let $\Omega$
 be a probability space with probability measure $P$ and let $\Phi=\Phi(\omega), \omega\in \Omega,$
 be an $n\times N$ random matrix and suppose $\Delta(\omega)$ is a corresponding family of
 decoders. We say that $(\Phi(\omega),\Delta(\omega))$ is {\em $(2,2)$ instance optimality in probability } of  order $k$ with constant $C$ if
 $$
 \|x-\Delta(\Phi x)\|_2 \leq C\sigma_k(x)_2
 $$
holds with high probability for this particular $x$.
 To state our results, we first introduce two properties that the  random matrix $\Phi$ should satisfy (see  \cite{jams1}).
 \begin{defn}
 We say that the $n\times N$ random matrix $\Phi$ satisfies RIP of order $k$
 with constants $\delta_k$ and probability
 $1-\epsilon$ if there is a set $\Omega_0\subset \Omega$ with $P(\Omega_0)\geq  1-\epsilon$ such that for all $\omega\in \Omega_0$ the matrix  $\Phi(\omega)$ satisfies
 $$
(1-\delta_k)\|x\|_2^2 \leq \|\Phi(\omega) x\|_2^2 \leq (1+\delta_k)\|x\|_2^2, \quad
\text{ for all } x\in \Sigma_k.
 $$
 \end{defn}
 \begin{defn}\label{de:bound}
 We say that the random matrix $\Phi$ has the boundedness property
 with constant $C$ and probability $1-\epsilon$ if for each $x\in
 \R^N$ there is a set $\Omega_0(x)\subset \Omega$ with $P(\Omega_0(x))\geq 1-\epsilon
 $ such that for all $\omega \in \Omega_0(x)$,
 \begin{equation}\label{eq:bound}
 \|\Phi(\omega)x\|_2^2 \leq C\|x\|_2^2.
 \end{equation}
 \end{defn}
The key thing in Definition \ref{de:bound} is that the set of $\omega$'s where
(\ref{eq:bound}) holds depends on $x$.
 The both properties above have been shown for random matrices of the Gaussian or Bernoulli type (see \cite{bound}).
Now we can state the result of instance optimality  with probability for $\ell_2$ norm.
 \begin{thm}\label{th:l2}
 Suppose $0<\delta \leq 1$ and $\alpha:=\lceil16+15\delta \rceil$.
 Assume that $\Phi$ is a random matrix which satisfies the RIP
 condition   $\delta_k+(1+\delta)\delta_{\alpha k} \leq \delta$ with probability  $1-\epsilon$, and also satisfies the boundedness
property with  constant $C$ and probability $1-\epsilon$. Then for each $x\in  \R^N$, there exists a set $\Omega(x)\subset \Omega$ with $P(\Omega(x))\geq
 1-2\epsilon$ such that for all $\omega \in \Omega(x)$ and
 $\Phi=\Phi(\omega)$,
 $$
 \|x^*-x\|_2 \leq C_3 \sigma_k(x)_2
 $$
 where $x^*:=\OMP_{(\alpha-1)k}(\Phi x)$ and $C_3=1+\sqrt{C(1+\delta)}(1+\sqrt{11+20\delta})$.
 \end{thm}
\begin{remark}
According to Theorem \ref{th:l2}, OMP can achieve (2,2) instance optimality in
probability of order $k$ after at least $15k$ steps. Suppose that $\Phi$
satisfy the conditions in Theorem \ref{th:l2}. Then, the following question is
interesting: {\em what is the minimal value of $c_0$ for which OMP can achieve
(2,2) instance optimality in probability of order $k$ after $c_0k$ iterations?}
\end{remark}

 \section{The proofs of Theorem  \ref{th:lplq} and Theorem  \ref{th:l2}}

 We first recall a result obtained by T. Zhang, which plays an important role in our analysis.
 \begin{thm}(\cite{zhang})\label{th:zhang} Let $\bar x\in \R^N$ and $0<\delta \leq 1$.
If the RIP condition $\delta_{\|\bar x\|_0}+(1+\delta)\delta_{\alpha\|\bar x\|_0} \leq \delta $ holds, then when $s=(\alpha-1) \|\bar x\|_0$, we have
 $$
 \|\Phi x^*-y\|_2^2 \leq (11+20\delta)\|\Phi \bar x-y\|_2^2,
 $$
 where  $x^*=\OMP_s(y)$ and  $\alpha=\lceil16+15\delta \rceil$.
 \end{thm}
We also need the following lemmas:
 \begin{lem}\label{le:1}
 Let $\Phi$ be any matrix which satisfies $\RIP(L,\delta)$ with $ \delta<1$ and
 $$
 L:=k\left(\frac{N}{k}\right)^{2-\frac{2}{p}}.
 $$
 Then for any $z\in \R^N$ and $1\leq p <2$ we have
 $$
 \|\Phi z\|_2\leq \sqrt{1+\delta}(\|z\|_2+\|z\|_p/k^{1/p-1/2}).
 $$
 \end{lem}
 \begin{proof}
 Let $T_1$ denote the set of  indices of the $L$ largest entries of $z$, $T_2$ the next $L$ largest, and so on. The last set $T_h$ defined this way may have less than $L$ elements. Set $T_0:=\cup_{j=2}^h T_j$. Then
 \begin{eqnarray}
 \|\Phi z\|_2 &\leq& \|\Phi z_{T_1}\|_2+\|\Phi z_{T_2}\|_2+\cdots+\|\Phi z_{T_h}\|_2 \nonumber\\
 &\leq & \sqrt{1+\delta}(\|z_{T_1}\|_2+\|z_{T_2}\|_2+\cdots+\|z_{T_h}\|_2).\label{eq:le1}
 \end{eqnarray}
 Since for any $i\in T_{j+1}$ and $i'\in T_j$, we have $|z_i|\leq |z_{i'}|$, which implies that
 $$
 |z_i|^p\leq L^{-1}\|z_{T_j}\|_p^p.
 $$
Then we obtain that
 $$
 \|z_{T_{j+1}}\|_2 \leq L^{1/2-1/p} \|z_{T_j}\|_p,\quad j=1,\ldots,h-1.
 $$
Noting $h-1\leq N/L$, we obtain that
 \begin{eqnarray}
 & &\sum_{j=2}^h\|z_{T_j}\|_2 \leq L^{1/2-1/p}\sum_{j=1}^{h-1}\|z_{T_j}\|_p \nonumber\\
 &\leq&  L^{1/2-1/p} (h-1)^{1-1/p}\|z\|_p\leq  \|z\|_p/k^{1/p-1/2}.\label{eq:le2}
 \end{eqnarray}
 Then (\ref{eq:le1}) and (\ref{eq:le2}) give
  $$
 \|\Phi z\|_2 \leq \sqrt{1+\delta}(\|z\|_2+\|z\|_p/k^{1/p-1/2}).
 $$
 \end{proof}

\begin{lem}(\cite{jams1})\label{le:cohen}
 Set $r:=1/p-1/q$ and suppose $r\geq 0$. Then
 $$
 \sigma_k(z)_{q}\leq \|z\|_p/k^{r}.
$$
 \end{lem}

We now have all ingredients to prove our conclusion:
 \begin{proof}[Proof of Theorem \ref{th:lplq}]  We first consider the $(2,p)$ case with $1\leq p<2$. Taking  $\bar x:=\beta_{2k}(x)$ in Theorem \ref{th:zhang}, we obtain that
 \begin{eqnarray*}
 \|\Phi x^*-y\|_2 &\leq &\sqrt{11+20\delta}\|\Phi \beta_{2k}(x)-\Phi x-e\|_2\\
 &\leq& \sqrt{11+20\delta}(\|\Phi \beta_{2k}(x)-\Phi x\|_2+\|e\|_2),
 \end{eqnarray*}
 where $x^*:=\OMP_{2(\alpha-1) k}(\Phi x+e)$.
Noting that $\# \supp(x^*-\beta_{2k}(x))\leq 2\alpha k$ and ${1}/{\sqrt{1-\delta_{2\alpha k}}}\leq \sqrt{1+\delta}$, we have
 \begin{eqnarray}\label{eq:pr1}
 & &\|x^*-\beta_{2k}(x)\|_2\nonumber \\
  &\leq& \sqrt{1+\delta} \|\Phi x^*-\Phi \beta_{2k}(x)\|_2\nonumber\\
  &\leq& \sqrt{1+\delta} (\|\Phi x^*-y\|_2+\|\Phi \beta_{2k}(x)-y\|_2)\nonumber\\
  &\leq & \sqrt{1+\delta}(\sqrt{11+20\delta}+1)(\|\Phi \beta_{2k}(x)-\Phi x\|_2+\|e\|_2) \nonumber\\
  &\leq &{(1+\delta)}(\sqrt{11+20\delta}+1) (\sigma_{2k}(x)_2+\sigma_{2k}(x)_p/k^{1/p-1/2}+\|e\|_2),
 \end{eqnarray}
 where the last inequality uses Lemma \ref{le:1}.
 We now consider the term $\sigma_{2k}(x)_2$ which appears on the right side of (\ref{eq:pr1}).  Lemma \ref{le:cohen} gives
 \begin{eqnarray}\label{eq:midpr}
 \sigma_{2k}(x)_2=\sigma_k(x-\beta_k(x))_2 \leq \|x-\beta_k(x)\|_p/k^{1/p-1/2}=\sigma_k(x)_p/k^{1/p-1/2}.
 \end{eqnarray}
 Combining (\ref{eq:pr1}) and (\ref{eq:midpr}), we obtain that
 \begin{equation}\label{eq:6}
  \|x^*-\beta_{2k}(x)\|_2 \leq 2{(1+\delta)}(\sqrt{11+20\delta}+1) (\sigma_{k}(x)_p/k^{1/p-1/2}+\|e\|_2).
 \end{equation}
Equations (\ref{eq:midpr}) and (\ref{eq:6}) imply that
 \begin{eqnarray}\label{eq:pr2}
 & &\|x^*-x\|_2\nonumber\\
  &\leq& \sigma_{2k}(x)_2+\|x^*-\beta_{2k}(x)\|_2\nonumber\\
  &\leq& \sigma_k(x)_p/ k^{1/p-1/2}+\|x^*-\beta_{2k}(x)\|_2 \nonumber\\
 &\leq & \left({2(1+\delta)}(\sqrt{11+20\delta}+1)+1\right) (\sigma_k(x)_p/k^{1/p-1/2}+\|e\|_2).
 \end{eqnarray}

We next consider the general case.  We first recall that H\"{o}lder inequality, which says
$$
(|a_1|^u+\cdots+|a_k|^u)^{1/u}(|b_1|^v+\cdots+|b_k|^v)^{1/v} \geq |a_1b_1|+\cdots+|a_kb_k|
$$
provided $1/u+1/v=1$ and $u,v>0$ where $a,b\in \R^k$. Then, by  H\"{o}lder inequality, when $1\leq q<2$, we have
$$
\left((|b_1|^q)^{2/q}+\cdots+(|b_k|^q)^{2/q}\right)^{q/2}\cdot k^{1-q/2}\geq |b_1|^q+\cdots + |b_k|^q,
$$
which implies that
\begin{equation}\label{eq:hold}
\|b\|_2\,\, \geq\,\, \frac{\|b\|_q}{k^{1/q-1/2}}.
\end{equation}
It follows from expressions (\ref{eq:midpr}) and
(\ref{eq:pr2}) that the condition
 \begin{eqnarray*}
 \|x^*-\beta_{2k}(x)\|_2 &\leq &\|x^*-x\|_2+\sigma_{2k}(x)_2 \\
 &\leq&C'(\sigma_k(x)_p/k^{1/p-1/2}+\|e\|_2)+\sigma_k(x)_p/k^{1/p-1/2} \\
 &\leq & (C'+1)\sigma_k(x)_p/k^{1/p-1/2}+C'\|e\|_2,
 \end{eqnarray*}
 where $C'={2(1+\delta)}(\sqrt{11+20\delta}+1)+1$.
 Note that $\#\supp(x^*-\beta_{2k}(x))\leq 2\alpha k$. The inequality (\ref{eq:hold}) provides the bound
 $$
 \|x^*-\beta_{2k}(x)\|_2\geq \frac{\|x^*-\beta_{2k}(x)\|_q}{(2\alpha k)^{1/q-1/2}}.
 $$
 Then, combining inequalities above, we arrive at
 \begin{equation}\label{eq:budeng}
 \|x^*-\beta_{2k}(x)\|_q \leq (2\alpha)^{1/q-1/2} (C'+1) k^{1/q-1/p}\sigma_k(x)_p+C'(2\alpha k)^{1/q-1/2}\|e\|_2.
 \end{equation}
  From Lemma \ref{le:cohen} and expression (\ref{eq:budeng}), we obtain the relation
 \begin{eqnarray*}
 & &\|x^*-x\|_q \\
  &\leq& \|x-\beta_{2k}(x)\|_q+\|x^*-\beta_{2k}(x)\|_q=\sigma_{2k}(x)_q+ \|x^*-\beta_{2k}(x)\|_q\\
  &\leq & \sigma_k(x)_p/k^{1/p-1/q}+(2\alpha)^{1/q-1/2} k^{1/q-1/p}(C'+1)\sigma_k(x)_p+C'(2\alpha k)^{1/q-1/2}\|e\|_2 \\
  &\leq &(1+(2\alpha)^{1/q-1/2}(C'+1))\frac{\sigma_k(x)_p}{k^{1/p-1/q}}+C'(2\alpha)^{1/q-1/2} k^{1/q-1/2} \epsilon,
 \end{eqnarray*}
  where $C'={2(1+\delta)}(\sqrt{11+20\delta}+1)+1$.
 The conclusion follows.
 \end{proof}

 \begin{proof}[Proof of Theorem \ref{th:l2}]
 We build  the proof on the ideas of Cohen, Dahmen and DeVore \cite{jams1}. We give the full proof for completeness.

Using the triangle inequality, we have
 \begin{equation}\label{eq:triangle}
 \|x-x^*\|_2 \leq \|x-\beta_k(x)\|_2+\|\beta_k(x)-x^*\|_2=
 \sigma_k(x)_2+\|\beta_k(x)-x^*\|_2.
 \end{equation}
 Let  $\Omega_0$ and $\Omega(x-\beta_k(x))$  be, respectively, the set in the definition of RIP in  probability and the set in the definition of
 boundedness in probability for the vector $x-\beta_k(x)$.
 We  set $\Omega':=\Omega_0\cap \Omega(x-\beta_k(x))$. A simple observation
 is $P(\Omega')\geq 1-2\epsilon$.
Also, Theorem \ref{th:zhang} implies that
$$
\|y-\Phi x^*\|_2 \leq \sqrt{11+20\delta}\|y-\Phi
\beta_k(x)\|_2.
$$
Noting that ${1}/{\sqrt{1-\delta_{2\alpha k}}}\leq \sqrt{1+\delta}$,
for any $\omega \in \Omega'$, we obtain  that
 \begin{eqnarray*}
 \|\beta_k(x)-x^*\|_2 &\leq& \|\Phi(\beta_k(x)-x^*)\|_2/\sqrt{1-\delta_{\alpha k}}\leq \sqrt{1+\delta } \|\Phi(\beta_k(x)-x^*)\|_2\\
 &\leq & \sqrt{1+\delta }(\|y-\Phi x_T\|_2+\|y-\Phi x^*\|_2) \\
 &\leq &\sqrt{1+\delta}(1+\sqrt{11+20\delta})\|\Phi x-\Phi
 \beta_k(x)\|_2\\
 &\leq & \sqrt{C}
 \sqrt{1+\delta}(1+\sqrt{11+20\delta}) \|x-\beta_k(x)\|_2
 \\
 &=& \sqrt{C} \sqrt{1+\delta}(1+\sqrt{11+20\delta}) \sigma_k(x)_2.
 \end{eqnarray*}
 Here, the first inequality uses the RIP and $\# {\rm supp}(\beta_k(x)-x^*)\leq \alpha k$ and the last inequality uses the boundedness property in probability for $x-\beta_k(x)$.
 Combining (\ref{eq:triangle}) and the equation above, we have
 $$
\|x-x^*\|_2\leq \left(1+\sqrt{C(1+\delta)}(1+\sqrt{11+20\delta}) \right) \sigma_k(x)_2.
 $$
 \end{proof}

\bigskip \medskip

\noindent {\bf Authors' addresses:}

\medskip

\noindent Zhiqiang Xu,
LSEC, Inst.~Comp.~Math., Academy of Mathematics and Systems Science,
 Chinese Academy of Sciences, Beijing, 100091, China.
 {\tt Email: xuzq@lsec.cc.ac.cn}

\end{document}